\documentclass{svproc}
\usepackage{graphicx}%
\usepackage{multirow}%
\usepackage{amsmath,amssymb,amsfonts}%
\usepackage{mathrsfs}%
\usepackage[title]{appendix}%
\usepackage{xcolor}%
\usepackage{textcomp}%
\usepackage{manyfoot}%
\usepackage{booktabs}%
\usepackage{algorithm}%
\usepackage{algorithmicx}%
\usepackage{algpseudocode}%
\usepackage{listings}%
\usepackage{subcaption}
\usepackage[hyphens]{url}
\usepackage{url}

\begin{document}
\mainmatter              % start of a contribution
\title{The Effect of Expertise Signaling when Sharing Medical Information in Opinion Dynamics}

\titlerunning{The Effect of Expertise Signaling in Opinion Dynamics}  % abbreviated title (for running head)
\author{Moyi Li$^{\star}$\inst{1} \and Chuqiao Yan$^{\star}$\inst{2}
 Jesus Jair Reyes Gutierrez \inst{3} \and  Aleksandra Urman \inst{2} \and Stefania Ionescu \inst{3}}
 
\authorrunning{M. Li \& C. Yan et al.} % abbreviated author list (for running head)

\institute{TU Delft, \and 
University of Zurich,
% \email{chuqiao.yan@uzh.ch},
% \\ WWW home page:\texttt{http://users.com/author1.html}
\and ETH Zurich}

\maketitle              % typeset the title of the contribution
\renewcommand{\thefootnote}{\fnsymbol{footnote}}
\footnotetext[1]{These authors contributed equally to this work.}
% \footnotetext[4]{Author's contribution to this work was primarily conducted at University of Zurich}
\renewcommand{\thefootnote}{\arabic{footnote}}

\begin{abstract}
Social media heavily influences health decisions, yet few opinion dynamics models account for how information is delivered. This work investigates how signaling source expertise shapes trust and long-term opinion formation.
Through an online experiment, we find that adding a professional title to the sender's display name significantly increases a post's perceived credibility. Incorporating these findings to create a data-driven opinion dynamics model, we prove that (1) connecting users through some paths to an expert guarantees consensus with experts' opinion, while (2) higher rates of expertise signaling accelerate reaching consensus.
Simulations show that (3) convincing even 20\% of users to signal their title reduces consensus time to nearly one-third in Erd\H{o}s--R\'enyi networks and (4) higher title adoption brings average opinion much closer to that of experts in scale-free and real-world networks.
Altogether, our work highlights expertise signaling as a viable intervention to promote validated medical information online.

\end{abstract}
\section{Introduction}
\label{sec: introduction}
%%%%%%%%%%%Older version%%%%%%%%%%%
% Background
Social media has become an important source for people seeking health-related information. For example, social media played a key role in the dissemination of treatment progress and health guidance during the COVID-19 pandemic \cite{zhong2021risk,song2025social}. Beyond the pandemic, social media has also been used to seek information about specific diseases and to engage in health-related discussions \cite{taik2024use,zhou2023circadian}. Information shared on social media shapes people's knowledge and opinions, and may affect their decisions \cite{ahmed_social_2019}. When related to health, online information could influence self-care, mental state, and lead people to take actions without consulting professionals \cite{thapa2021influence,oh2017investigating}. Therefore, trusting and sharing wrong information can have particularly serious consequences \cite{tasnim2020impact}. 
% Trusting and promoting medical information or misinformation on social media can lead to serious, sometimes even life-threatening consequences. Social media's ability to rapidly disseminate information exacerbates this issue.

It is important to understand what affects people's trust in health information on social media. Prior studies focus on analyzing target audience groups and political relations \cite{li_seeking_2018,de_choudhury_seeking_2014,urman2022politicization}. Less research has examined the relationship between the specific features of post senders and the credibility of their posts. Previous research has demonstrated that people's trust in information sources online is associated with the perceived expertise of the sources \cite{wathen_believe_2002}. For example, \cite{bass_effects_1986} found that clinical therapists' professional titles could influence their perceived credibility. \cite{park2020trust} also suggests that displaying `Dr.' or `PhD' on Reddit increases feedback in expert communities. However, these works either do not focus specifically on medical information or do not consider online information dissemination. To address this, we investigate whether signaling expertise (`Dr.', `PhD') can increase trust in medical information on social media.
% Bass's study about titles focused on trust offline, leaving the effects of expertise signaling online unknown.
%instead of in the online environment, such as social media. The long-term effects of trust in online interactions have not been fully explored either.
% Whether titles, as long as other source features, can influence the perceived credibility of medical information spreading online remains undiscovered.

While short-term effects can be investigated via online experiments, opinion dynamics models are uniquely positioned to understand how these effects compound over time. Examples include the DeGroot model \cite{degroot_reaching_1974},  the Friedkin-Johnsen (FJ) model \cite{friedkin1990social}, and the Flocking model \cite{asch_opinions_1955}. %In this study, we will use the DeGroot model to study trust-weighted opinion updates and incorporate individual stubbornness following the FJ model. 
While such models have been extensively used in the Complex Networks community, recent work has increasingly emphasized the need for data-driven models that better reflect real-world behavior by incorporating empirical findings. Doing so improves the realism and explanatory power of the model in opinion propagation, belief dynamics, customer behavior, and fairness in digital systems \cite{bell2018data,dalege2024networks,du2022data,rieder2025critical,wang2025data}. Aligning with this vision, we propose a data-driven agent-based model grounded in the DeGroot and FJ models with weights calibrated based on data from the online experiment. By varying the behavior of agents~\cite{gilbert2005simulation,hu2017competing,schweighofer2020agent}, we then test the effects of interventions such as encouraging more experts to signal their expertise.

%To simulate opinion dynamics on social media, we then build an data-driven Agent-based Model (ABM) based on results from  online experiment in the first step. Agent-based modeling has been widely used for studying opinion dynamics as it offers flexibility to integrate heterogeneous agent attributes, interaction rules, and interventions \cite{gilbert2005simulation,hu2017competing,schweighofer2020agent}. Beyond the theory-driven ABM, recent work increasingly emphasizes data-driven agent-based modeling that better reflects real-world behavior by incorporating empirical findings. Incorporating empirical data, data-driven ABM improves the realism and explanatory power of the model in opinion propagation, belief dynamics, customer behavior fairness in digital systems \cite{bell2018data,dalege2024networks,du2022data,rieder2025critical,wang2025data}. In the simulation, we investigate how signaling expertise when creating health-related content affects the credibility of this content and the long-term opinion dynamics on social media platforms.

Altogether, our work investigates how signaling expertise when creating health-related content affects the credibility of this content and the long-term opinion dynamics on social media platforms. To do so, we use a data-driven agent-based modeling approach consisting of two steps. First, we conduct an exploratory online experiment to study the impact of different creator features related to expertise signaling on content credibility. The results show that including professional titles, such as `Dr.' or `MD', in the display name makes receivers trust content more. Second, we leverage the data from the online experiment to calibrate an opinion dynamics model. We use both theoretical analysis and agent-based simulations\footnote{Our code is publicly available on GitHub (\url{https://github.com/LiMY00/Expert-Title-in-Opinion-Dynamics})\label{githublink}} to understand the long-term implications of different expertise-signaling attitudes on the opinion dynamics. The results suggest that platform interventions encouraging experts to signal their expertise could more quickly promote information agreed upon as true among experts.

\section{Empirical Study: Online Experiment}
%%%%%%%%%%%Updated Version%%%%%%%%%%%
%In this section, we are going to show the empirical study we conducted online. The study design and its procedure is conducted in Section \ref{subsec: study design} and Section \ref{subsec: participants}. Section \ref{subsec: data analysis and results} shows the data analysis and its results.

\textbf{Study Design.}
We conducted a between-subject exploratory online experiment and considered four factors: sender’s username title (three levels: Non-title or `Dr.' or `MD'), sender gender (two levels: male or female), profile information that indicates relevant medical background (two levels: relevant or irrelevant), and the presence of a source link in the message content (two levels: with or without a link to a scientific source supporting the statement).
To examine the effects of these factors, we designed the experiment in three blocks. In Block 1, Gender (female vs. male) and Title (no title vs. `Dr.' vs. `MD') were varied, resulting in six experimental conditions. In Block 2, Source link (included vs. excluded) and Title (no title vs. `Dr.' vs. `MD') were varied, also resulting in six conditions. In Block 3, only profile information (irrelevant vs. relevant) was varied, resulting in two conditions. Participants were randomly assigned to one condition within each block. Fig.~\ref{fig:factor_wrokflow} shows an example of the measured factors and the experiment workflow.

\begin{figure}[t]
\centering
\begin{subfigure}[t]{0.38\columnwidth}
    \centering
    \includegraphics[width=\linewidth]{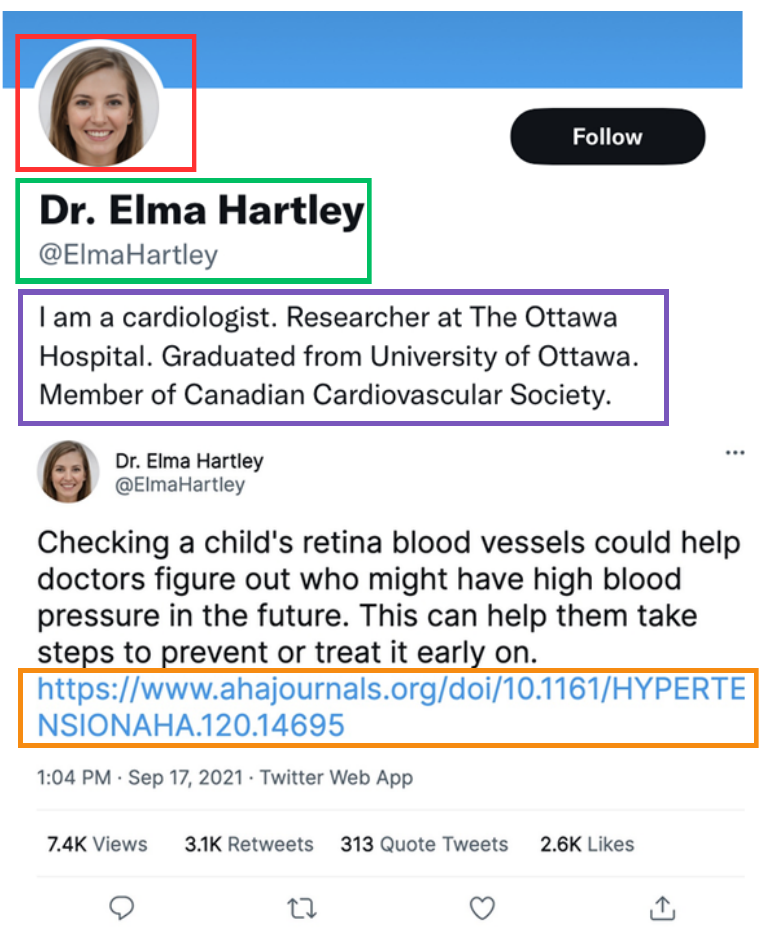}
    \caption{Factors examined in X posts.}
    \label{fig:factors}
\end{subfigure}
\hfill
\begin{subfigure}[t]{0.58\columnwidth}
    \centering
    \includegraphics[width=\linewidth]{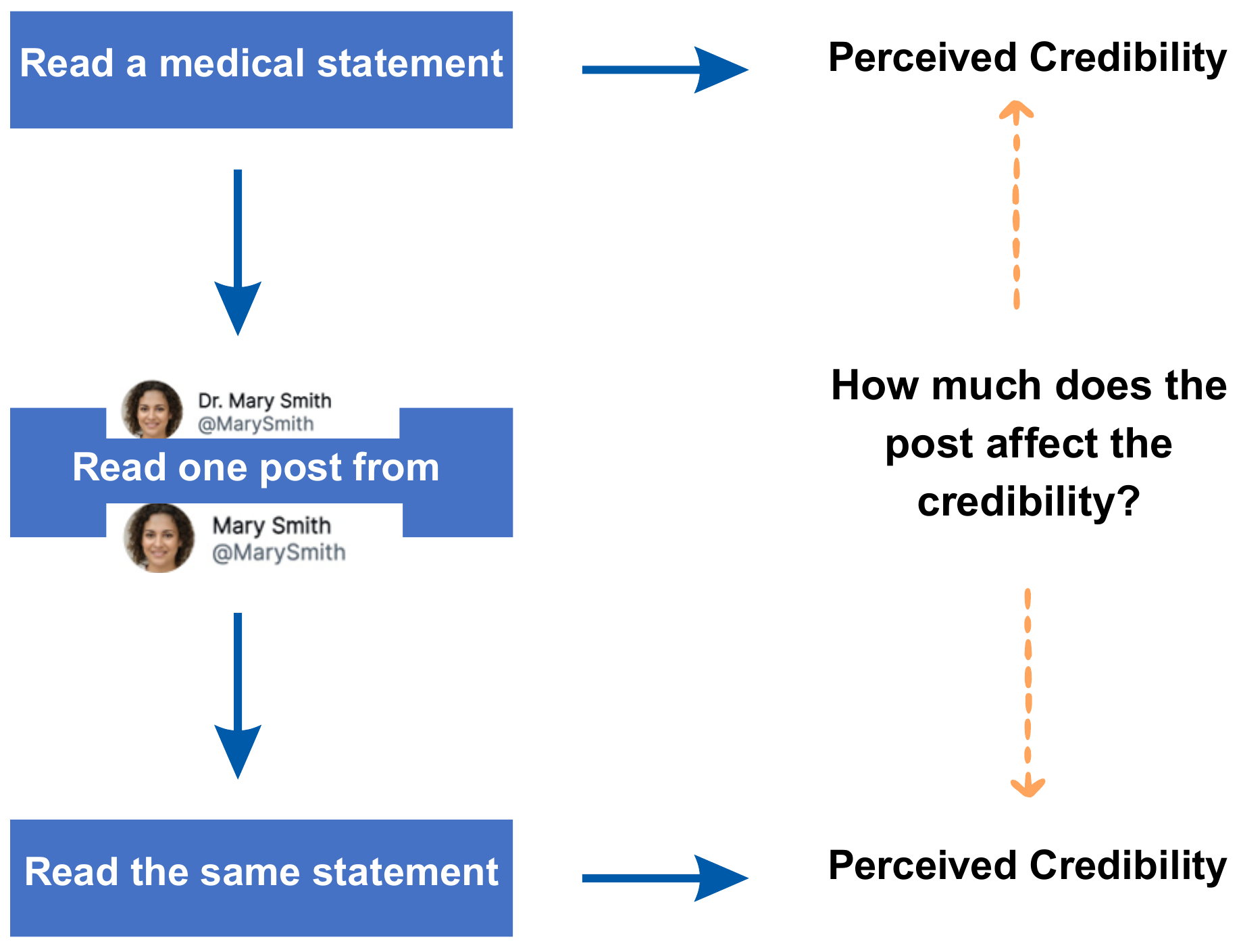}
    \caption{Online experiment workflow.}
    \label{fig:workflow}
\end{subfigure}
\caption{Overview of the factors examined in X posts and the experimental workflow. (a): Example X post showing four measured factors: the profile image at the top of the post, the professional title directly below the profile name, the profile description in the
middle, and the source link below the post text. (b): Workflow of the online experiment. Each participant was randomly assigned to one experimental block and rated the perceived credibility of two posts.}
\label{fig:factor_wrokflow}
\end{figure}

\textbf{Participants.} 
We recruited 180 participants (90 female, 86 male, 4 preferred not to say) through \textit{Prolific}\footnote{\url{https://www.prolific.com}}. The age ranges from `under 20’ to `over 60’, and most participants were aged 20–30 (69\%). These participants were located in 27 countries and completed the experiment in English. Each participant received €0.75 upon completion and passing the attention check.

\textbf{Data Analysis and Results.}
Since the independent variable, \textit{Title}, has three groups, we used a one-way ANOVA to compare the changes in perceived credibility across the three \textit{Title} conditions (no title vs. `Dr.' vs. `MD'). The analysis revealed a significant effect of \textit{Title}, ($F(2,177)=9.86$, $p < .001$). Moreover, we conducted Tukey's HSD post hoc tests to determine which groups differed. The results showed that both the `Dr.' and `MD' groups differed significantly from the no-title group ($p=.0002$ and $p=.0016$, respectively), whereas the two professional-title groups did not differ significantly ($p=.95$) in the credibility they induced. These results indicate that posts from users with a professional title in their display name were perceived as more credible than those from users without a title. Table \ref{tab:anova_tukey} summarizes these results.

\iffalse\begin{table}[t]
\centering
\caption{ANOVA and Tukey HSD post-hoc results for the effect of signaled title of sender in their display name on perceived credibility.}
\label{tab:anova_tukey}
\setlength{\tabcolsep}{4pt}
\begin{tabular}{lrrrrr}
\toprule
\multicolumn{6}{l}{\textbf{ANOVA}} \\
\midrule
\textbf{Source} & \textbf{df} & \textbf{Sum Sq.} & \textbf{Mean Sq.} & \textbf{F} & \textbf{$p$} \\
\hline
Title    & 2   & 21.361  & 10.680 & 9.865 & $<.001$ \\
Residual & 177 & 191.639 & 1.083  &       &         \\
\toprule
\multicolumn{6}{l}{\textbf{Tukey HSD post-hoc comparisons}} \\
\hline
\textbf{Comparison} & \multicolumn{3}{r}{\textbf{Mean Diff.}} & \multicolumn{2}{r}{\textbf{$p$-value}} \\
\hline
`Dr.' vs. No title & \multicolumn{3}{r}{0.748} & \multicolumn{2}{r}{$<.001$} \\
MD vs. No title  & \multicolumn{3}{r}{0.686} & \multicolumn{2}{r}{$.002$}  \\
`Dr.' vs. `MD'     & \multicolumn{3}{r}{0.062} & \multicolumn{2}{r}{$.946$}  \\
\bottomrule
\end{tabular}
\end{table}\fi

\begin{table}[t]
\centering
\caption{ANOVA and Tukey HSD post-hoc results for the effect of the sender’s signaled title in their display name on perceived credibility.}
\label{tab:anova_tukey}
\setlength{\tabcolsep}{8pt}
\begin{tabular}{l *{5}{r}}
\toprule
\multicolumn{6}{l}{\textbf{ANOVA:}} \\

\textbf{Source} & \textbf{df} & \textbf{Sum Sq.} & \textbf{Mean Sq.} & \textbf{$F$} & \textbf{$p$} \\
\midrule
Title    & 2   & 21.361  & 10.680 & 9.865 & $< .001$ \\
Residual & 177 & 191.639 & 1.083  &       &          \\
\midrule [\heavyrulewidth]
\addlinespace[0.5em]
\multicolumn{6}{l}{\textbf{Tukey HSD Post-Hoc Comparisons:}} \\

\multicolumn{3}{l}{\textbf{Comparison}} & \multicolumn{2}{r}{\textbf{Mean Diff.}} & \textbf{$p$-value} \\
\midrule
\multicolumn{3}{l}{`Dr.' vs. No title} & \multicolumn{2}{r}{0.748} & $< .001$ \\
\multicolumn{3}{l}{`MD' vs. No title}  & \multicolumn{2}{r}{0.686} & $.002$  \\
\multicolumn{3}{l}{`Dr.' vs. `MD'}       & \multicolumn{2}{r}{0.062} & $.946$  \\
\bottomrule
\end{tabular}
\end{table}

For the remaining variables, we used  Welch's two-sample t-tests to examine their impact on the perceived credibility, namely \textit{Gender}, \textit{Profile}, and \textit{Link}. The results suggested no significant difference in the change of perceived credibility caused by the sender’s gender, relevance of profile information, or link to scientific support in the content ($p_{gender} = 0.42$, $p_{link} = 0.34$, $p_{profile} = 0.17$). %We include anonymized data and the code for the statistical analysis in our GitHub repository \footnote{https://github.com/LiMY00/Trust-in-Medical-information-Online-Data-driven-Agent-based-Modeling.}.

\section{Data-driven opinion dynamics model}
The above data analysis indicates that including professional titles in users' display names has a small yet significant impact on the credibility of their posts. A natural question is how this effect compounds within a social network of multiple users and affects opinion dynamics. To answer this, we extend classical opinion dynamics models to account for expertise signaling. We use the data from the online experiment to inform the trust parameters used in opinion updates.

\paragraph{Preliminaries}
We denote a social network by a directed graph $G = (V, E)$, where $V$ is the set of \textit{agents} (social media users) and $E$ is the set of edges (following relationships). At every iteration (timestep), each agent $i$ holds an opinion $x_i(t)\in [0, 1]$ capturing their level of trust in a given statement (with $1$ corresponding to full trust). Opinions update via a DeGroot model~\cite{degroot_reaching_1974}:
\begin{equation} \label{equ:degroot}
x(t + 1) = W x(t),
\end{equation}
where $W$ denotes the row-stochastic adjacency matrix of the network $G$, with $w_{ij}$ denoting how much agent $i$ trusts agent $j$'s opinion. When $w_{ii} = 1$, agent $i$ always preserves their opinion, and we refer to it as \emph{stubborn}.

\iffalse Thus, for node $i$, the updating rule at time $t\in \mathbb{N}$ is:

\begin{equation} \label{equ:degroot_individual}
x_{i}(t+1)=\frac{w_{ii}x_{i}(t)+\sum _{j \in N(i)}w_{ij}x_{j}(t)}{w_{ii}+\sum _{j \in N(i)}w_{ij}}, 
\end{equation}
where $N(i)$ denotes the set of neighbors of node $i$, $w_{ij}$ denotes the how much would agent $i$ trust agent $j$'s opinion. \fi

\paragraph{User types.}
We extend the above model by distinguishing between common users and experts (Table \ref{tab:agent_type}) that differ in stubbornness, title, and initial opinion. More precisely, \textit{common users} are the general population on the social media platform who update their opinions at each timestep. \textit{Experts} hold a single opinion universally accepted as true among them and could either signal their expertise or not. We track expertise signaling as our experiment showed it can affect trust\footnote{As there was no significant difference between `Dr.' and `MD' revealed from the online experiment, we only capture one general form of expertise signaling (`Dr.').}.
% \textit{Malicious users} refer to users who intentionally disseminate false or misleading information to deceive public perception and cause harm, such as trolls, bots, or impersonators. \textit{Influencers} are users who possess substantial followings, thus we define them as the top $2\%$ of users who own most of the followers. 
% The stubbornness and title do not change over interaction. The opinion adapts based on rules in the following subsection within $[0,1]$, where opinion = 1 represents the agent completely trusts an opinion, and opinion = 0 represents they do not trust this opinion at all.
\begin{table}[t]
\centering
\caption{Overview of the different types of agents in our model. For each type, we specify: if the agent is \textit{stubborn} (i.e., keeps their initial opinion fixed), the professional \textit{title} they could include in their display name, and the range of values for their \textit{initial opinion}.}
\label{tab:agent_type}
\small
\begin{tabular}{p{2.8 cm} p {2.1 cm}  p {2.5 cm} p {2.8 cm}}
\hline
Agent type & Stubborn & Title & Initial opinion\\
\hline
\hline
Common user & False & None & $[0, 1]$\\
Expert & True & `Dr.'/ None & 1\\
% Malitious user & True & `Dr.'/ None & 0\\
% Influencer & True & `Dr.'/ None & 1\\
\hline
\end{tabular}
% \iffalse
% \begin{tabular}{cccc}
% \hline
% \textbf{\begin{tabular}[c]{@{}c@{}}Agent \\ Type\end{tabular}} & \textbf{\begin{tabular}[c]{@{}c@{}}Stubbor-\\ nness\end{tabular}} & \textbf{Title}                                                   & \textbf{\begin{tabular}[c]{@{}c@{}}Initial \\ opinion\end{tabular}} \\ \hline
% \begin{tabular}[c]{@{}c@{}}Common \\ user\end{tabular}         & False                                                             & Non-title                                                        & {[}0,1{]}                                                           \\ \hline
% Doctor                                                         & True                                                              & \begin{tabular}[c]{@{}c@{}}'`Dr.'' title/\\ Non-title\end{tabular} & 1                                                                   \\ \hline
% \begin{tabular}[c]{@{}c@{}}Malicious \\ user\end{tabular}      & True                                                              & \begin{tabular}[c]{@{}c@{}}'`Dr.'' title/\\ Non-title\end{tabular} & 0                                                                   \\ \hline
% Influencer                                                     & True                                                              & Non-title                                                        & 1                                                                   \\ \hline
% \end{tabular} \fi
\end{table}

\paragraph{Assumptions and parameters.} %Experiment-based Parameter Calibration
While both our model and code (GitHub repository mentioned above) for the simulation are general and could capture a variety of settings, we make some assumptions to limit the scope of this work. The model assumes experts unanimously agree on one opinion (without loss of generality, the one with value $1$) and there is only one type of expertise signaling. In terms of information dynamics, the model assumes that the common user on social media trust other actors than experts equally, and there is no marginal effect of time (i.e., the impact of message from the same type of user is constant at every time step).  %As usual, when we refer to malicious users, we assume they try to promote the opposite opinion ($0$) and are stubborn~\cite{egele2013compa}.
The model above is based on a few parameters, which we overview in Table \ref{tab:sim_parameter}. For the simulations, we ground the parameter choices in our online experiment data and prior work. %We provide an overview below and leave the details for Section Simulation in the Appendix.

% \begin{table}[ht]
% \caption{\label{tab:virt_exp} Tabular description of model parameters (Left), the values taken in the simulations (Center), and the virtual experiments for which these values are varied (Right). We use bold (red) for default parameter values: i.e., we use this value when the parameter in question is not varied. For each combination of parameters, we run $100$ simulations.} \center
% \label{tab:sim_parameter}
% \small
% \begin{tabular}{p{4.3 cm} p {3.5 cm}  p {2 cm}}
% \hline
% Varied Parameters & Values Taken & Vary in\\
% \hline
% \hline
% \% of experts with title & $0, 20, \dots, \textbf{\color{red}100}$ & E1, E4\\
% Number malicious users & \textbf{\color{red}0}, 100, $\dots$, 400& E2\\
% \% titled malicious users &  $\textbf{\color{red}0}, 25, \dots, 100$ & E2\\
% \% titled influencers & $\textbf{\color{red}0}, 25, \dots, 100$  & E3\\
% Network & ER (edge pr. 0.01), & E1, E2\\
%  & \textbf{\color{red}SF} ($\alpha = 2.5$) & \\
%  & PolBlog & \\
% % \hline
% \hline
% Fixed Parameters & Value & \\
% \hline
% \hline
% Total number of agents & 1500 \cite{adamic_political_2005} & \\
% Number of experts & 100 & \\
% Opinion of experts & 1 & \\
% Opinion malicious users & 0 & \\
% Number of influencers & 30 (top-degree 2\%) & \\
% % Network & Directed\&Weighted & \\
% \hline
% \end{tabular}
% \end{table}

\begin{table}[ht]
\caption{\label{tab:sim_parameter} Tabular description of model parameters (Left) and the values taken in the simulations (Right). We use bold (red) for default parameter values: i.e., we use this value when the parameter in question is not varied. For each combination of parameters, we run $100$ simulations with different random seeds.}
\center
\small
\begin{tabular}{p{4.3cm} p{6.5cm}}
\hline
Varied Parameters & Values Taken \\
\hline
\hline
\% of experts with title & $0, 20, \dots, \textbf{\color{red}100}$ \\
Network & ER (edge pr. 0.01), \textbf{\color{red}SF} ($\alpha = 2.5$), PolBlog \\
\hline
\hline
Fixed Parameters & Value \\
\hline
\hline
Total number of agents & 1500 \cite{adamic_political_2005} \\
Number of experts & 100 \\
Opinion of experts & 1 \\
\hline
\end{tabular}
\end{table}

\begin{itemize}
    \item \emph{Initial opinion} values are grounded in participants’ pre-experiment trust in the medical statements. We first map the 7-point Likert scale to $[0,1]$, yielding $\mathcal{N}(0.54, 0.29)$. Then, we sample common users’ initial opinions from this distribution, truncated to $[0,1]$.

    \item \emph{Weight matrix} values are derived from participants' trust in posts from accounts with different display names. For participant $i$, the experiment corresponds to a single-step DeGroot update after seeing a post from an expert: \begin{equation} x_i(1)=w^*_{ii}x_i(0)+w^*_{ij}x_j(0), \quad w^*_{ii}+w^*_{ij}=1, \quad x_j(0)=1
    \end{equation} where $x_i(0)$ and $x_i(1)$ denote pre- and post-experiment opinions of participant $i$ and , and $x_j(0)$ denotes the sender's opinion. Solving for $w^*_{ij}$ and averaging by sender type gives $w^*_{i,Dr}=0.183, w^*_{i,CU}=0.003.$
    % We averaged the data in groups of non-title and '`Dr.'' title, and gained the influence of doctors $I_{d}=0.183$, the influence of non-title users $I_{c}=0.003$. According to Equ. (\ref{equ:influence}), the influence here is only in contrast to the self-influence of participants' opinions from last step\footnote{As we only studied one time step in the online experiment, the self-influence here may include both the influence of last time step and the stubbornness discussed in the Friedkin-Johnsen model.}. 
    For simulations, we follow the approach of \cite{zhou_two-step_2020} and preserve row stochasticity by normalizing the trust based on each agent's connectivity.

    \item \emph{Networks.} 
    We use the Erd\H{o}s--R\'enyi (ER) and Barabási–Albert models \cite{newman2001random} to create ER networks and scale-free (SF) networks \cite{mislove2007measurement,eikmeier2017revisiting}. We also use a real-world political blog network (PolBlog) \cite{adamic_political_2005}, which is widely used in research on information diffusion \cite{peralta2021effect,peralta2022opinion,zhang2025optimal}.
    % These reflect the typical connectivity of real-world social media platforms, where a few high-degree nodes (influencers) dominate visibility and diffusion. 
    % However, \citeauthor{broido2019scale} study also suggests that real-world social networks are mostly weak scale free \cite{broido2019scale}. This network 
\end{itemize}

%\paragraph{Assumptions.} 
%While both our model and code\footnote{https://github.com/LiMY00/Trust-in-Medical-information-Online-Data-driven-Agent-based-Modeling.} for the simulation are general and could capture a variety of settings, we make some assumptions for the scope of this work. More precisely, we assume experts unanimously agree on one opinion (without loss of generality, the one with value $1$). As usual, we assume malicious users try to promote the opposite opinion \cite{egele2013compa}. % In one of the experiments, we assume all influencers assist in promoting the verified information. Additionally, when simulating the effect of RSs, we assume all doctors have already signaled titles. 
%We also assume that there are more bots than doctors on social media \cite{ng2025global}.

\section{Expertise Signaling Impacts Opinion Dynamics}
\paragraph{Evaluation metrics.}
To evaluate the social impact of signaling expertise on social media platforms, we adopt standard metrics and investigate (1) whether a \textit{consensus} (i.e., a state in which opinions no longer change and all agents hold the same opinion) is reached, and, if so, (2) the \textit{time needed to reach consensus} and (3) the \textit{consensus value}. These metrics are crucial as they capture whether communities of users will eventually agree, what they would agree on, and how long it takes for this agreement to arise. In the context of medical information, consensus among scientists and public consensus around expert-backed evidence is socially beneficial, as it promotes informed decision-making and reduces the influence of misinformation \cite{van2015highlighting,maibach2016importance}.
%, we examine three aspects in the simulation experiments for different scenarios:  \textit{i)} the \textbf{final opinions distribution} of common users,  \textit{ii)} the trajectory of the \textbf{average opinion} over time, and  \textit{iii)} the \textbf{speed} at which opinions change.  
%When consensus is reached, we analyze the final average opinion and the time of convergence. 
However, not all discrete-time linear systems converge to a consensus \cite{degroot_reaching_1974,berger1981necessary,dong_managing_2017} and many real-life platforms remain polarized \cite{schmidt2018polarization,jiang2021social}. Finally, we consider the \textit{average user opinion} over time.

\paragraph{Increased connectivity to experts promotes convergence to the experts' opinion.}
As Theorem \ref{theorem_result1} demonstrates, ensuring a directed path between common users and experts guarantees that all users reach consensus on the verified experts' opinion. 
Consequently, the average user opinion also converges to this exact value.
In practice, this suggests that increasing connectivity to experts (either by encouraging following relationships or by directly intervening on recommendation systems to promote their message) drives the network toward a consensus where all users adopt the verified opinion.

\begin{theorem}[Existence of consensus]
\label{theorem_result1}
Consider a DeGroot system $x(t+1) = W x(t)$ with $n$ agents and a non-empty set of stubborn agents holding the same opinion, $c$. % , partitioned into a set of non-stubborn agents $\mathcal{N}$ and a non-empty set of stubborn experts $\mathcal{S}$ holding some opinion $c$. 
If there exists a directed path from every non-stubborn agent to at least one expert, the system converges to a global consensus where both individual beliefs and the average user opinion equal $c$, i.e., $\lim_{t \to \infty} x(t) = c \cdot \mathbf{1}_n \implies \lim_{t \to \infty} \frac{1}{n} \mathbf{1}_n^\top x(t) = c$.
\end{theorem}

\begin{proof}[Sketch]
%\jai{\noindent \textit{Proof Sketch}
Let $\mathcal{N}$ be the set of non-stubborn agents and $\mathcal{S}$ be the non-empty set of stubborn experts holding some opinion $c$.
Drawing upon established results in algebraic graph theory \cite{bullo2019lectures}, we reorder the agents to partition the row-stochastic matrix $W$ into block upper-triangular form: $W = \begin{bmatrix} W_{\mathcal{N}\mathcal{N}} & W_{\mathcal{N}\mathcal{S}} \\ \mathbf{0} & I \end{bmatrix}$. 
Here, $W_{\mathcal{N}\mathcal{N}}$ captures the trust of non-stubborn agents towards other non-stubborn agents, $W_{\mathcal{N}\mathcal{S}}$ denotes the trust from non-stubborn agents towards experts, and $I$ reflects that experts maintain their beliefs. 
The path condition guarantees that the matrix $W_{\mathcal{N}\mathcal{N}}$ is sub-stochastic; thus, its spectral radius is bounded strictly below $1$, ensuring asymptotic convergence. 
Solving for the equilibrium point of the non-stubborn agents as per  
Discrete-time Affine Systems (\cite{bullo2019lectures}, p.35), resolves to $c \cdot \mathbf{1}_{\vert{}\mathcal{N}\vert{}}$. 
Because every individual agent's opinion converges to $c$, the mean of the agent's opinions converges to $c$ as $t \to \infty$. \hfill $\blacksquare$
\end{proof}

\paragraph{An intervention encouraging more experts to signal their expertise speeds convergence.}
Theorem~\ref{theorem_result2} looks at the speed of convergence. It states that if experts become more trusted (i.e., their corresponding weights increase), then consensus is reached faster. In connection with the results from the online experiment, this suggests that encouraging more experts to signal their expertise will speed the time needed to reach consensus.

\begin{theorem}[Impact of weights on the convergence speeds]
\label{theorem_result2}
Consider a DeGroot system, $x(t+1) = W x(t)$, satisfying the conditions of Theorem~\ref{theorem_result1}. Let a second influence matrix $W'$ be formed by increasing the trust of non-stubborn agents in stubborn agents. Then, the convergence speed under $W'$ is at least as fast as under $W$.
Furthermore, if at least one weight is strictly larger (i.e., $W'>W$), then the new system converges faster.
\end{theorem}

\begin{proof}[Sketch]
Building on the block partition $W = \begin{bmatrix} W_{\mathcal{N}\mathcal{N}} & W_{\mathcal{N}\mathcal{S}} \\ \mathbf{0} & I \end{bmatrix}$ introduced in Theorem \ref{theorem_result1}, an increase in the trust weights of the block ($W_{\mathcal{N}\mathcal{S}}$) requires a corresponding decrease in the trust weights of the block ($W_{\mathcal{N}\mathcal{N}}$) to maintain the row-stochastic property of $W$. 
According to the Monotonicity of Spectral Radius, (Lemma 4.10 \cite{bullo2019lectures}, p.56), this decrease implies that the new sub-matrix $W'_{\mathcal{N}\mathcal{N}}$ satisfies $\rho(W'_{\mathcal{N}\mathcal{N}}) \le \rho(W_{\mathcal{N}\mathcal{N}})$, with a strict inequality if at least one weight strictly decreases. 
Under modal decomposition, the evolution of the system's state can be expanded into its eigenmodes as $x(t) = \sum \lambda_i^t (v_i^\top x_0)v_i$. 
Because the system's transient states decay proportionally to the eigenvalue powers $\lambda_i^t$, and these non-consensus eigenvalues are bounded by $\rho(W_{\mathcal{N}\mathcal{N}})$ and smaller than $1$, a reduced spectral radius accelerates their decay to zero. 
Consequently, both the individual agent states and the average user opinion converge to the expert consensus at a faster rate. \hfill $\blacksquare$
\end{proof}

%\jai{It is worth noting that the theorems presented above can be naturally extended to setups with malicious agents. 
%Consider a DeGroot system $x(t+1) = W x(t)$ partitioned into three disjoint sets: non-stubborn agents $\mathcal{N}$, a non-empty set of stubborn experts $\mathcal{E}$ holding opinion $c$, and a non-empty set of stubborn malicious users $\mathcal{M}$ holding opinion $d$, where $c>d$. 
%By reordering the agents, the row-stochastic matrix $W$ expands into a $3 \times 3$ block structure:
%$$W = \begin{bmatrix} W_{\mathcal{N}\mathcal{N}} & W_{\mathcal{N}\mathcal{E}} & W_{\mathcal{N}\mathcal{M}} \\ \mathbf{0} & I_{\vert{}\mathcal{E}\vert{}} & \mathbf{0} \\ \mathbf{0} & \mathbf{0} & I_{\vert{}\mathcal{M}\vert{}} \end{bmatrix}$$

%where $W_{\mathcal{N}\mathcal{N}}$ captures trust of non-stubborn agents towards other non-stubborn agents, $W_{\mathcal{N}\mathcal{E}}$ and $W_{\mathcal{N}\mathcal{M}}$ denotes the trust from non-stubborn agents towards experts and malicious users, respectively, and $I_*$ reflects the static beliefs of both groups.}

\paragraph{Convergence time and average opinion greatly improve as more experts signal their expertise.}
%In Fig. \ref{Fig:E1}, we plot the dynamics of average opinion in 100 simulations for both ER and SF networks to observe the speed for convergence (if any) and the distribution of final opinion. 
While the analytical results specify the direction of change resulting from the intervention, the magnitude of change remains unknown. Therefore, we use simulations to find the intervention's impact across settings. Fig. \ref{Fig:E1} (a) shows that on ER networks agents reach a consensus to a final opinion of $1$. Moreover, the time for reaching consensus %(respectively $2022.76$, $762.74$, $582.29$, $522.59$, $496.55$, $470.64$ steps) 
greatly reduces from the setting in which no doctors signaling their expertise to the one when $20\%$ do so (from 2023 to 763 steps). Further increasing this percentage produces diminishing gains in convergence speed.
%The results of E1 on the SF networks tend to be random even within the same condition. 
In Fig. \ref{Fig:E1} (b), we plot only three conditions on the SF network: $0\%, 50\%, 100\%$ doctors have `Dr.' displayed in usernames. Although the average opinions increase when more doctors signal titles, the confidence intervals overlap. By investigating the simulation runs, we believe this is caused by the randomness in the position of experts within the network. More precisely, the gains would be large if the `Dr.' signaling expertise has a high degree and negligible if they are a periphery node.
We plot the same conditions for the PolBlog network in Fig.~\ref{Fig:E1} (c). Although the opinions do not converge to $1$ as on the ER network, there is less variance in average opinion across runs compared to the SF network. As with the ER network, signaling expertise is especially effective when the first doctors do so.

\emph{Policy implications.} Overall, these results suggest that encouraging doctors to signal expertise is an efficient intervention. Across scenarios, it can bring the average opinion of users closer to that of experts. In networks with high homogeneity and uniformly random connections, it can greatly reduce the time people need to be fully convinced of the experts' opinion. Such an intervention is particularly effective when only a few experts already signal their expertise.

\emph{Accounting for malicious users and influencers.} Having a modular code allowed us to investigate additional scenarios and interventions. First, we assumed we have malicious users, such as bots, promoting information opposite to the experts~\cite{egele2013compa}. In this case, we found it is more important to ensure these users cannot pose as experts (e.g., by verifying accounts that include a title) than to reduce their number through moderation. Second, we investigated the efficacy of increasing the connectivity of experts. To this end, encouraging influencers (i.e., high degree nodes) to re-post content from verified experts improves the average opinion.
Moreover, in accordance with the analytical results, we found that a recommender system that boosts the visibility of even one verified expert guarantees convergence and greatly improves the average opinion over time. % even beyond influencers' reposting. 
We include details of this additional work in the repository\footnote{\url{https://github.com/LiMY00/Expert-Title-in-Opinion-Dynamics}} and plan to further investigate these directions in future work. Meanwhile, these results demonstrate the potential of our data-driven opinion dynamics model to support policymakers and social media platforms in prioritizing interventions.
% \textcolor{red}{One paragraph about other simulation results: Our agent-based model also supports additional interventions concerning expertise. Examples include verifying accounts signaling expertise, inviting influencers to promote experts' content, or increasing the visibility of experts' posts. More simulation code and results for scenarios involving misinformation and visibility allocation can be found in the GitHub link above\footref{githublink}, which is a potential extension work in the future.}

% \begin{figure*}[t]
%     \centering
%     \begin{minipage}[t]{0.30\linewidth}
%         \centering
%         \includegraphics[width=\linewidth]{Figures/E1_ER.png}
%         \caption*{(a) ER network }
%     \end{minipage}
%     \hspace{0.02\linewidth}
%     \begin{minipage}[t]{0.30\linewidth}
%         \centering
%         \includegraphics[width=\linewidth]{Figures/E1_SF.png}
%         \caption*{(b) SF network}
%     \end{minipage}
%     \hspace{0.02\linewidth}
%     \begin{minipage}[t]{0.30\linewidth}
%     \centering
%     \includegraphics[width=\linewidth]{Figures/E1_PBN.png}
%     \caption*{(c) PolBlog network}
%     \end{minipage}
%     \caption{E1 - The average opinion over time for different networks and different shares of experts signaling expertise. Lines show the average over 100 simulations. The shaded areas represent the standard deviation in average opinion across simulations.}   
% \label{Fig:E1}
% \end{figure*}

 \begin{figure*}[t]
    \centering
    \begin{subfigure}[t]{0.32\linewidth}
        \centering
        \includegraphics[width=\linewidth]{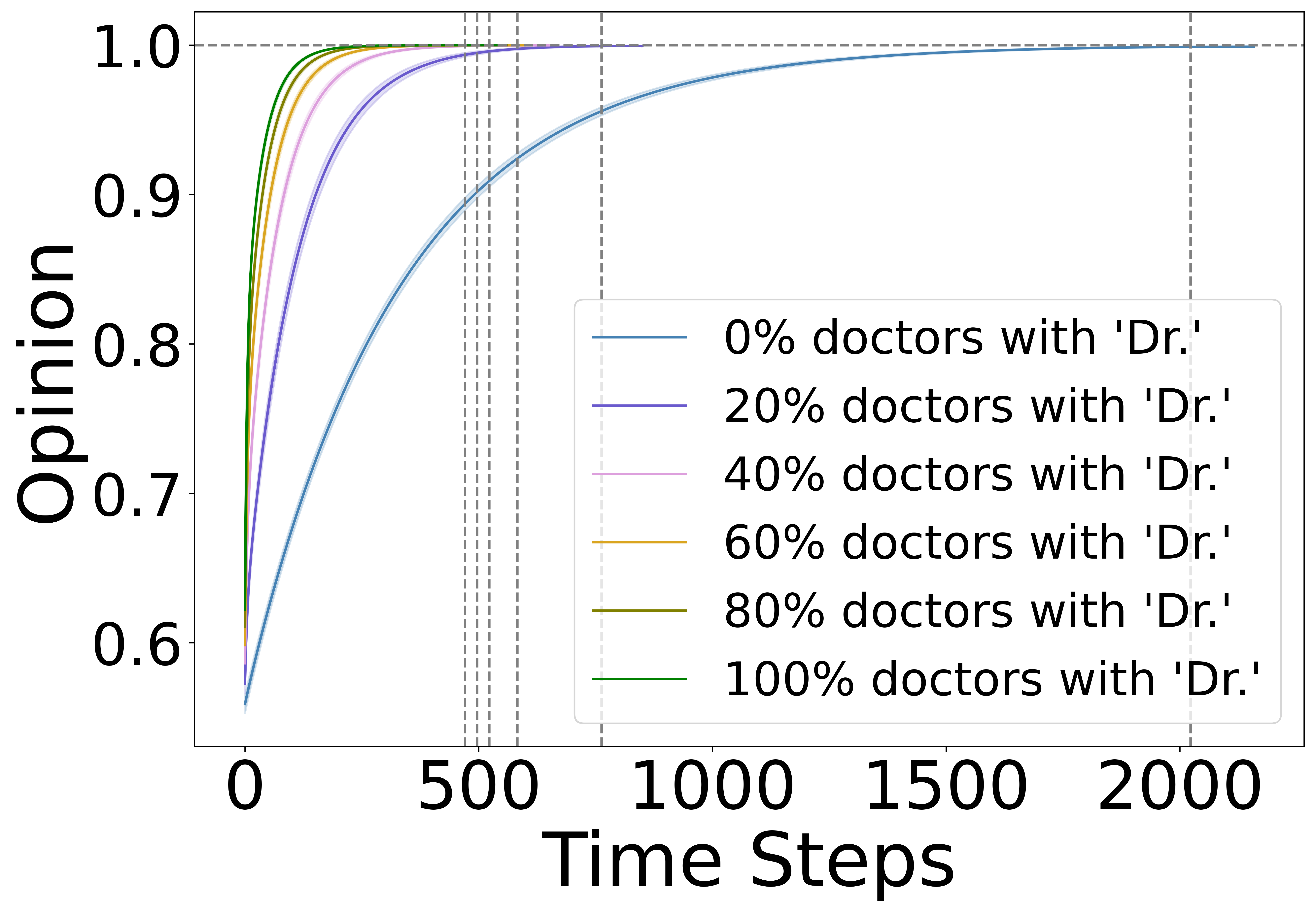}
        \caption{ER network}
    \end{subfigure}
    \hfill
    \begin{subfigure}[t]{0.305\linewidth}
        \centering
        \includegraphics[width=\linewidth]{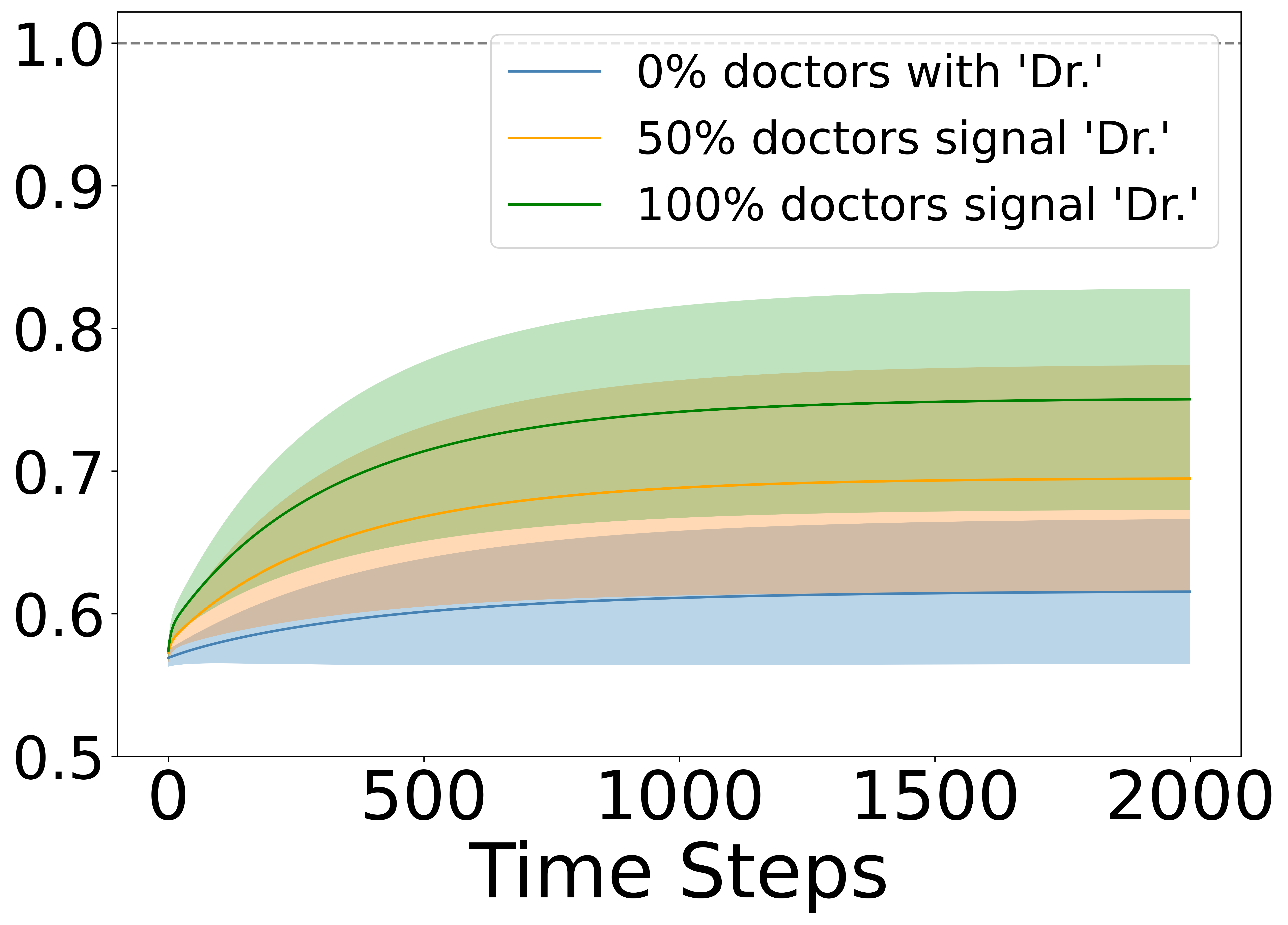}
        \caption{SF network}
    \end{subfigure}
    \hfill
    \begin{subfigure}[t]{0.305\linewidth}
        \centering
        \includegraphics[width=\linewidth]{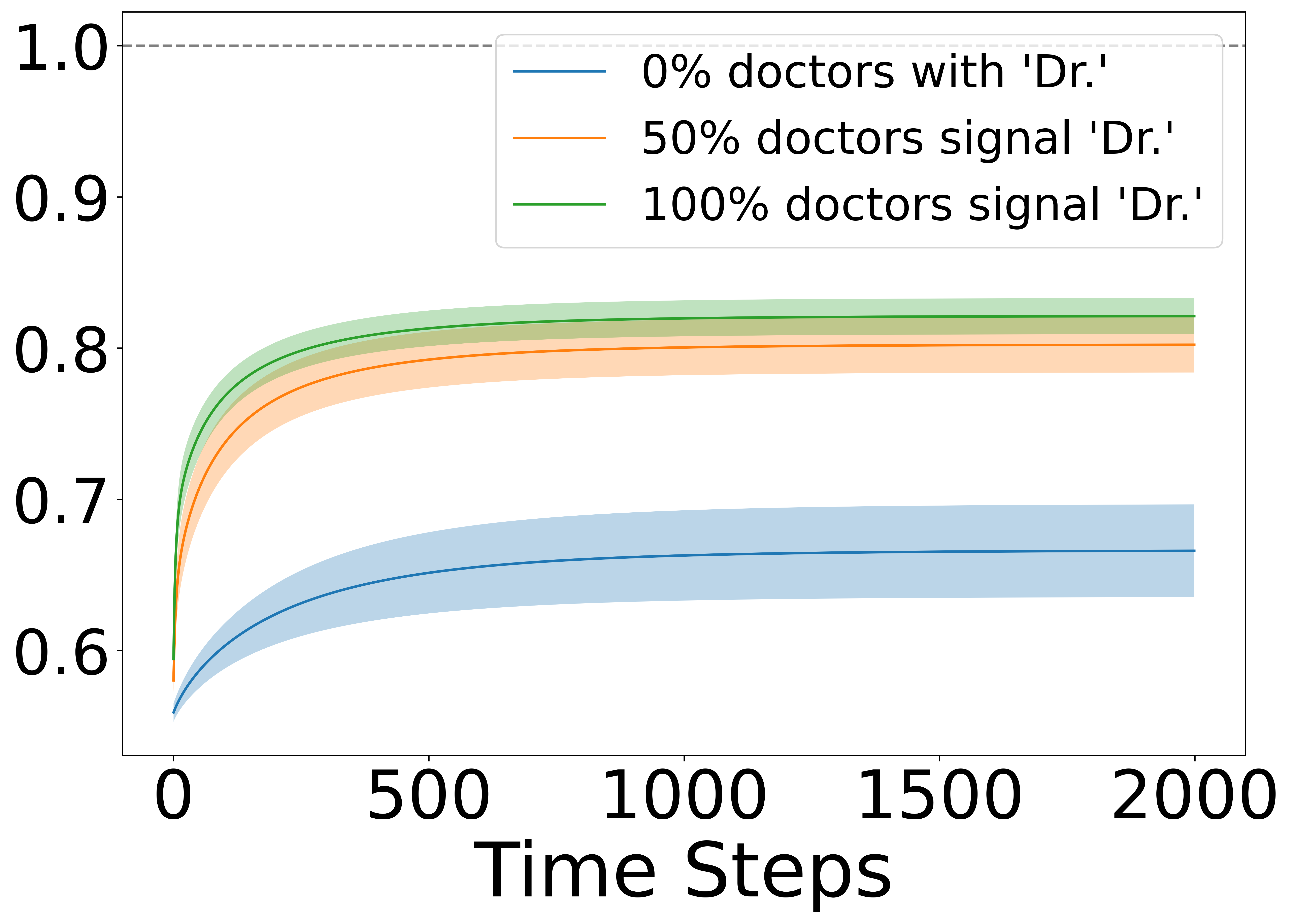}
        \caption{PolBlog network}
    \end{subfigure}
    \caption{Average opinion over time for different networks and shares of experts signaling expertise. Lines show the average over 100 simulations. The shaded areas represent the standard deviation in average opinion across simulations.}
    \label{Fig:E1}
\end{figure*}

\section{Conclusion}
We investigate how signaling expertise (e.g., including professional titles) shapes trust in medical information and influences long-term opinion dynamics on social networks. Through an online experiment, we find that including professional titles (e.g., `Dr.', `MD') in the display name of the source slightly yet significantly increases the perceived credibility of their posts. We use this data to extend the DeGroot model with expertise signaling by grounding its parameters in the experiment's data. A theoretical analysis of the resulting discrete-time system specifies sufficient conditions for consensus to emerge and demonstrates that encouraging expertise signaling increases the speed of reaching consensus. Through simulations, we find that on both random (ER) and real-world (PolBlog) networks, even a small group of experts who include their title can greatly and significantly improve the average opinion. While results on scale-free networks are not significant, the trends persist.
%we find that on a random network, even a small group of titled users can have a large influence on the aggregate opinion. On a scale-free network, however, simply encouraging experts to signal titles or verifying identity (e.g., removing malicious users or their `Dr.' titles) may be less effective than interventions through influencers (e.g., encouraging them to promote verified medical information). %We also compared interventions of inviting influencers to promote and improving experts' visibility through recommender systems, providing insights into which strategy to prioritize under different conditions.
%Results also suggest that an even more effective intervention could be to algorithmically promote the content of experts (e.g., through recommender systems).

With this being a first exploratory analysis, several limitations remain to be tackled by future research.
% Nonetheless, our work remains an initial step and has several limitations that could be tackled by future research. %Being an initial step, several limitations remain to be addressed by future work. 
For example, %our empirical experiment is exploratory. 
future work can conduct a power analysis and build a confirmatory experiment; such work can also verify if our findings apply to more diverse samples of participants. Like any agent-based model, our model makes several simplifying assumptions on user behavior (e.g., binary stubbornness, equal level of trust in non-expert sources, static trust networks). Additionally, we did not consider the time and cost of implementing interventions (e.g., it can be time-consuming for experts to actively participate on social media~\cite{galetti2017scientists}, limiting intervention integration).
% Nonetheless, our work remains an initial step and has several limitations that could be tackled by future research. We do not distinguish users' level of stubbornness, nor do we test and implement how users trust influencers versus ordinary users differently. The model also assumes static networks and does not capture diminishing effects of repeated exposure to expertise signals. Additionally, user behavior is simplified, omitting interactions like likes, comments, and reposts. Finally, our empirical experiment was exploratory, thus small in scale. Future work shall validate these findings with larger, more diverse samples and explore content moderation strategies.

%More generally, we believe our work puts an important pillar in the study of expertise signaling on social media. 
Altogether, our work is a first step towards understanding how expertise signaling and related interventions affect medical opinions. % of the general public.
At a methodological level, it combines online experiments, theoretical analysis, and agent-based modeling, demonstrating the benefits of a holistic approach in linking short-term effects to long-term impact. Having a dedicated online experiment allowed us to vary one variable at a time, thus isolating its effect. This represents an alternative to data-driven modeling methods that calibrate parameters based on digital trace data, which are more contextual and have many hard-to-disentangle effects. It also shows how isolating directions of change via dedicated experiments supports a more insightful theoretical analysis, connected to real-life changes and interventions.
%simple models can use both analytical methods (providing guarantees for simplified settings) and simulations (for investigating more complex, realistic settings and long-term societal effects). 
Thus, more generally, our work demonstrates the potential of experiment-driven modeling for supporting realistic social simulations and an initial comparison over potential interventions.

\section*{Acknowledgments}
This work was supported as a part of NCCR Automation, a National Centre of Competence in Research, funded by the Swiss National Science Foundation (grant number 51NF40\_225155). 
% \end{thebibliography}

% ---- OR ------- Uncomment this section to use bibtex
\bibliographystyle{spmpsci} % We choose the "plain" reference style
\bibliography{refs} % Entries are in the refs.bib file
\end{document}